\documentclass[12pt,a4paper]{article}

\usepackage{amsmath, amsfonts}
\usepackage{amssymb, graphicx}
\usepackage{amsthm}
\usepackage[english]{babel}
\usepackage{color}
\usepackage{hyperref}
\usepackage{epsfig}
\usepackage{amsmath}
\usepackage{natbib}
\usepackage{bm}
\usepackage{nccmath}
\usepackage{enumerate}
\usepackage{setspace}
\usepackage{doi}
\hypersetup{colorlinks,
            linkcolor=blue,
            anchorcolor=blue,
            citecolor=blue,urlcolor={blue}}

\newcounter{result}

\newtheorem{theorem}[result]{\textbf{Theorem}}

\newtheorem{proposition}[result]{\textbf{Proposition}}

\newtheorem{definition}{\textbf{Definition}}
\newtheorem{example}{\textbf{Example}}

\newtheorem{lemma}{\textbf{Lemma}}

\DeclareMathOperator*{\argmax}{argmax}

\begin{document}

\author{Ata Atay\thanks{Department of Mathematical Economics, Finance and Actuarial Sciences, University of Barcelona, and Barcelona Economic Analysis Team (BEAT), Barcelona, Spain. E-mail: \href{mailto:aatay@ub.edu}{aatay@ub.edu}} \and Eric Bahel\thanks{Department of Economics, Virginia Polytechnic Institute and State University, Blacksburg, VA 24061-0316, USA. E-mail: \href{mailto:erbahel@vt.edu}{erbahel@vt.edu}} \and Tam\'as Solymosi\thanks{Department of Operations Research and Actuarial Sciences, and Corvinus Centre for Operations Research, Corvinus University of Budapest, Budapest, Hungary. E-mail: \href{mailto:tamas.solymosi@uni-corvinus.hu}{tamas.solymosi@uni-corvinus.hu}}}

\title{Matching markets with middlemen \\under transferable utility\thanks{Ata Atay is a Serra H\'{u}nter Fellow (Professor Lector Serra H\'{u}nter) under the Serra H\'{u}nter Plan (Pla Serra H\'{u}nter). Ata Atay gratefully acknowledges financial support from the Spanish Ministry of Science and Innovation through grant PGC2018-096977-B-100. Tam\'{a}s Solymosi gratefully acknowledges financial support from the Hungarian National Research, Development and Innovation Office through grant NKFIH K-119930.}}
\maketitle

\begin{abstract}
This paper studies matching markets in the presence of middlemen. In our framework, a buyer-seller pair may either trade directly or use the services of a middleman; and a middleman may serve multiple buyer-seller pairs. Direct trade between a buyer and a seller is costlier than a trade mediated by a middleman. For each such market, we examine an associated TU game. First, we show that an optimal matching for a matching market with middlemen can be obtained by considering the two-sided assignment market where each buyer-seller pair is allowed to use the mediation service of the middlemen free of charge. Second, we prove that matching markets with middlemen are balanced. Third, we show the existence of a buyer-optimal and a seller-optimal core allocations. In general, the core does not exhibit a middleman-optimal allocation. Finally, we establish the coincidence between the core and the set of competitive equilibrium payoff vectors.

\end{abstract}

\newpage
\section{Introduction}

Consider a commodity whose market exhibits three types of agents: buyers, sellers, and middlemen. Each seller owns one indivisible unit; and each buyer seeks to purchase one unit (from any of the sellers) in exchange for money. Units may not be homogeneous, \textit{i.e.}, a buyer may have different valuations for the respective units owned by two distinct sellers.
We assume that utility is transferable between all agents; and this allows the use of cooperative games with transferable utility (or TU games, for short). A given buyer  and a given seller may trade directly, or they may use the services of a middleman. For example, in the real estate market, a seller may or may not use a realtor facilitating the sale of her house. In financial markets, brokers provide their service to investors  (in exchange for a fee); and each investor may or may not hire a broker. As is common in these applications, we assume that a middleman may serve multiple buyer-seller pairs. 

Markets with middlemen have been studied in  different contexts (search and matching models, general equilibrium model, etc.). \cite{rw87} was the first work  to study the activity of middlemen in search markets. In \cite{y94}  agents can search for matches on their own, or they can resort to a middleman who mediates between agents of opposite sides to facilitate their pairing. \cite{f97} investigates competition between middlemen when direct trade between buyers and sellers is available. He showed that direct trade has a negative effect on the market power of middlemen. \cite{br00} consider a market with search frictions and a monopolistic middleman where buyers and sellers bargain over the surplus. \cite{jl02} study a model in which sellers and buyers have heterogeneous tastes. They showed that middlemen are better off if they have a multi-unit inventory of differentiated products.

In their seminal paper, \cite{ss71} used a TU game to model a two-sided housing market where there are $m$ buyers and $n$ sellers. In their setting, each buyer is interested in  buying at most one house and each seller has one house for sale. Each buyer has $n$ valuations (one for each house) each seller has a reservation value for her house. The valuation matrix represents the joint surplus generated by each pair formed by a buyer and a seller. For this market situation the associated TU game, the so-called \emph{assignment game}, is defined. They studied a solution concept, the \emph{core}, which is the set of allocations that cannot be improved upon by any coalition. They showed that the core of an assignment game is always non-empty and has a lattice structure. Moreover, \cite{d82} and \cite{l83} proved that there exists a core allocation at which each buyer attains his/her marginal contribution to the grand coalition (the buyer-optimal core allocation) and there exists a core allocation at which each seller attains his/her marginal contribution to the grand coalition (the seller-optimal core allocation). 

Multi-sided matching markets  may in general have an empty core  under transferable utility (\citealp{kw82}). Thus, the remarkable results obtained for  two-sided markets cannot be generalized to all multi-sided markets. Several authors have examined conditions (on the structure of the market) allowing to show the non-emptiness of the core. \cite{s99} introduced a subclass of multi-sided matching markets where valuations are obtained from a supermodular function. She proved that any game in this subclass has a non-empty core. Some other authors have shown that matching markets exhibiting some additivity property have a non-empty core (see for instance \citealp{q91}; \citealp{t13}; \citealp{aetal16}). 

Among different multi-sided matching market models, there is a growing literature on matching markets with middlemen. \cite{s97} introduces a three-sided matching game with middlemen, the so-called \emph{supplier-firm-buyer game}. In this model, a buyer and a seller (supplier) can trade only through a middleman (firm). Hence, unlike our model, a mixed-pair of buyer-seller cannot generate any surplus without a middleman. The author showed that the class of supplier-firm-buyer game is balanced.

\cite{os14} consider a model of three-sided matching markets in which middlemen can mediate at most one trade between a buyer and seller. Buyers and sellers are allowed to trade directly as well as trade through a middleman. Unlike our model, middlemen incur a matching cost and moreover, the associated TU game only considers the matching situations with triplets of buyer-middleman-seller. Moreover, the present paper also relaxes the condition on middlemen by allowing each of them to serve  multiple buyer-seller pairs. 

In a recent paper, \cite{eom19} study a hybrid model of two-sided and multi-sided matching markets. They consider a two-sided model with buyers and sellers that are not disjoint. There exists a so-called central player who can act both as a buyer and as a seller. In their model, the central player has to be present for a trade between a buyer-seller pair. Otherwise, a trade cannot be realized. Hence, the central player has  veto power and, as explained by the authors, their model thus induces a veto game (\citealp{b16}). 

The present work takes a game-theoretical approach to matching markets with middlemen. We consider a class of three-sided matching market in which buyers and sellers can trade directly or indirectly through middlemen. Each seller owns an object to sell and each buyer wants to acquire at most one object. A trade between a mixed-pair of a buyer and a seller can be mediated by at most one middleman, meanwhile any given middleman can mediate trades between multiple buyer-seller pairs. Utility is transferable and quasi-linear in money. We assume that a direct trade between a buyer-seller pair is more costly (and therefore generates a lower surplus) than when a middleman is involved (thanks to her knowledge of the market, the middleman allows the buyer and the seller to lower their search costs and travel costs). Given a buyer-seller pair, the surplus generated by their exchange varies depending on the middleman serving them.  As mentioned before, every buyer (seller) can trade with at most one seller (buyer), whereas any given middleman may serve multiple buyer-seller pairs (note that any such pair is served by at most one middleman).

In order to study the core and its structure, we propose a simple procedure allowing to compute the worth of the grand coalition in any matching market with middlemen. Precisely, we construct an associated two-sided assignment market where the valuation of every buyer-seller pair is obtained by taking the maximum surplus that they can achieve either by a direct trade between themselves or by an indirect trade brokered by any of the middlemen in the market. In a similar fashion, \cite{mn11} introduce the maximum assignment game for a given collection of assignment games where any given coalition attains the maximum possible value among the given collection of games. However, in their case, the two authors observed that the maximum assignment game need not be an assignment game, and it may not even be superadditive.

Our main results are described as follows. First, we show that an optimal matching for a matching market with middlemen can always be constructed from an optimal matching of the associated two-sided market, and vice versa. Moreover, the maximum total surplus in the two markets are equal (Proposition \ref{prop:opt_matchings}). Second, we prove that the core of a market with middlemen is always non-empty by showing that the set of payoff vectors composed of a core allocation for the two-sided assignment market and zero payoffs to all middlemen is precisely the subset of the core of the market with middlemen where all middlemen payoffs are zero (Theorem \ref{TheoremCoreNonempty}). Furthermore, we prove that there exists a buyer-optimal allocation, that is, a core allocation that each and every buyer (weakly) prefers to all other core allocations. Likewise, there exists a seller-optimal core allocation (Theorem \ref{TheoremBuyOptimal}). Moreover, as in the standard two-sided model, our results guarantee that, at the buyer-optimal (seller-optimal) core allocations, each buyer (seller) achieves her  marginal contribution to the grand coalition. Interestingly, we provide an example  showing that, in general, there exists no middleman-optimal core allocation (Example \ref{ex:no_mc_middlemen}): all middlemen do not necessarily achieve their maximum core  payoffs simultaneously. Finally, we characterize the core in terms of competitive equilibrium payoffs (Theorem \ref{thm:Core_euqivalence_CE}).

The paper is organized as follows. In Section \ref{sec:prel} some preliminaries are given. Section \ref{sec:mode} introduces the model and explores the structure of its outcomes, the matchings. In Section \ref{sec:core} we prove the non-emptiness of the core and prove that there exists a side-optimal core allocation for the buyer side and also for the seller side in the market. In contrast, by means of an example we demonstrate that there need not exist a middleman-optimal core allocation. In Section \ref{sec:CE} we establish the coincidence between the core and the set of competitive equilibrium payoff vectors. Section \ref{sec:conc} concludes.

\section{Preliminaries}
\label{sec:prel}

A \textit{cooperative game with transferable utility}  (or TU game) is a pair  $(N,v)$ where $N$ is a non-empty, finite set of \textit{players (or agents)} and $v:2^{N}\rightarrow\mathbb{R}$ is a \textit{coalitional function} satisfying $v(\emptyset)=0$.  The number $v(S)$ is the \emph{worth} of the coalition $S\subseteq N$. Whenever no  confusion may arise as to the set of players, we will identify a TU game $(N,v)$ with its coalitional function $v$. 

Given a game $v$, a \textit{payoff allocation} (or allocation) is a tuple $x\in\mathbb{R}^{N}$ representing the players' respective allotments. The total payoff of a coalition $S\subseteq N$ is denoted by $x(S)=\sum_{t\in S} x_{t}$ if $S\neq \emptyset$ and $x(\emptyset)=0$. 

In a game $v$, an  allocation $x$ is called \emph{efficient} if $x(N)=v(N)$,  \textit{individually rational} if $x_{t}=x(\{t\})\geq v(\{t\})$ for all $t\in N$, and \textit{coalitionally rational} if $x(S)\geq v(S)$ for all $S\subseteq N$. The \textit{core} of $v$, denoted by $Core(v)$, is the set of coalitionally rational and efficient payoff allocations. A game is called \emph{balanced} if it has a non-empty core, and \emph{totally balanced} if all the subgames, i.e. the game restricted to the non-empty coalitions, are balanced. A totally balanced game $v$ is balanced and also \textit{superadditive}, i.e. $v(S\cup T)\geq v(S)+v(T)$ for any  coalitions $S,T\subseteq N$ such that $S\cap T=\emptyset$.

We call \textit{marginal contribution} of a player $t\in N$ in the game $v$ the quantity $mc_t(v)=v(N)-v(N\setminus \{t\})$.
It is well known that the marginal contribution is an upper bound of the payoffs attainable in the core for a player, i.e. $x_t\leq mc_t(v)$ for all $x\in Core(v)$ and $t\in N$, but this bound is not necessarily sharp.

\section{Matching markets with middlemen}
\label{sec:mode}

We consider a three-sided market where there are three disjoint sets of agents: the set of buyers $B=\{b_1,b_2,\ldots , b_I\}$, the set of middlemen $M=\{m_1, m_2, \ldots , m_J\}$, and the set of sellers $S=\{s_1,s_2,\ldots, s_K\}$. Note that the cardinalities $I,J,K$ of these respective sets  may differ. We call $B$ (or $S$) the \textit{short side of the market} if it holds that $I\leq K$ (or $|S|\leq |B|$). Let $N=B\cup M\cup S$ be the set containing all agents. In this market, each buyer-seller pair $(i,k)\in B\times S$ can trade directly with each other, or indirectly through some middleman $j\in M$ which results in a trade involving the triple $(i,j,k)\in B\times M\times S$. 

Each seller owns one unit of good and each buyer seeks to buy at most one unit of good.  Although a trade between each buyer-seller pair can be mediated by at most one middleman, any given middleman can mediate trades between multiple buyer-seller pairs. That is to say, each $j\in M$ can potentially serve the entire market by brokering as many trades as the cardinality of  the short side of the market. 

A market with middlemen can thus be described by specifying two non-negative matrices: (a) a two-dimensional  matrix $A=(a_{ik})_{\substack{i\in B \\ k\in S}}$ giving  the joint monetary \textit{surplus generated by every mixed pair} $(i,k)\in B\times S$ if they trade directly, and (b) a \textbf{three-dimensional} non-negative matrix $\hat{A}=(\hat{a}_{ijk})_{\substack{i\in B\\ j\in M \\k\in S}}$  representing the \textit{joint surplus generated by a trade between buyer $i$ and seller $k$ that is mediated by middleman $j$}. It will often be convenient to represent the three-dimensional surplus matrix $\hat{A}$ as an array of its two-dimensional layer submatrices indexed by the middlemen. Formally, $\hat{A}=[A^{(j)}\in \mathbb{R}^{B\times S}: j\in M]$ with elements $a^{(j)}_{ik}=\hat{a}_{ijk}$ for all $(i,j,k)\in B\times M\times S$. For sake of unified notation, we denote the two-dimensional surplus matrix $A$ as $A^{(0)}$ with elements $a^{(0)}_{ik}=a_{ik}$ for all $(i,k)\in B\times S$.

We will assume that 
\begin{equation}\label{Assmatrices}
a^{(j)}_{ik}=\hat{a}_{ijk}\geq a_{ik}=a^{(0)}_{ik}, ~\forall (i,j,k)\in B\times M\times S,
\end{equation} that is to say, a direct trade  between a buyer-seller pair  entails higher search costs than when  a middleman is involved. Hence, the total surplus of a trade with a middleman is greater than or equal to that of a trade without a middleman.

A \emph{market with middlemen} is fully described by a tuple of the type $\gamma=(B,M,S,A,\hat{A})$. Since the sets $B,M,S$ are  given and fixed, we will often describe such a  market with middlemen by simply specifying a pair of matrices $(A,\hat{A})$ satisfying (\ref{Assmatrices}).  

Call \textit{basic coalition} any subset of $N$ that is either a singleton $\{i\}$, or a pair $\{i,k\}$ such that $i\in B$ and $k\in S$, or a triple   $\{i,j, k\}$ such that $i\in B$, $j\in M$ and $k\in S$. Moreover, let  $\mathcal{B}^{N}=\{\{i,j,k\}\mid i\in B, j\in M, k\in S\}\cup\{\{i,k\}\mid i\in B, k\in S\}\cup \{\{i\}\mid i\in N\}$ be the collection of all basic coalitions. Furthermore, for all  $T\subseteq N$, denote by $\mathcal{B}^{T}$ the set of basic coalitions that have all their agents in $T$, that is,  $\mathcal{B}^{T}=\{E\in\mathcal{B}^{N} \mid E\subseteq T\}$.
Denote by $B_T$, $M_T$, and $S_T$ the set of buyers, middlemen, and sellers in coalition $T$, respectively. 

\begin{definition}\label{defmatching}
Given any $T\in 2^N$,  a collection  of basic coalitions $\mu$ will be called a $T$-\textit{matching} if it satisfies \textbf{(i)} $\mu\subseteq \mathcal{B}^{T}$; \textbf{(ii)} $B_T\cup S_T\subseteq \bigcup\limits_{E\in \mu} E$;   \textbf{(iii)} for any $t\in {B_T\cup S_T}$ and any distinct $E,F\in \mu$,  $t\notin E\cap F$; \textbf{(iv)} for all $j\in M_T$,   $\left[\{j\}\in \mu \right]\Rightarrow [j\notin E, \forall E\in \mu\setminus \{j\}]$.
\end{definition}

Remark that conditions (i)-(iii) in Definition \ref{defmatching} say that a buyer (seller) must belong to exactly one basic coalition in the collection $\mu$. It is possible for a middleman to belong to multiple basic triples of $\mu$ (since she may mediate multiple trades). However, as stated in (iv), a middleman  appearing in a singleton of $\mu$ should not belong to any other element of $\mu$. With a slight abuse of notation, we  write $k= \mu(i)$ and $i= \mu(k)$ for all $(i,k)\in B\times S$ such that $\left[\{i,k\}\in \mu \mbox{ or } \{i,j,k\}\in \mu \mbox{ for some } j\in M\right]$.
 We also write $\mu(t)=t$ for all $t\in N$  such that $\{t\}\in \mu$.
Let $\mathcal{A}(T)$ denote the set of $T$-matchings. 

Observe that a $T$-matching $\mu$ induces disjoint groups of buyer-seller pairs that trade via the same middleman. With a slight abuse of notation, we shorthand the subsets containing only one agent from each side of the market as an array in which the ordered specifies the type of the agents: $(i,k)$ means $\{i,k\}$ with $i\in B$ and $k\in S$; similarly, $(i,j,k)$ means $\{i,j,k\}$ with $i\in B$, $j\in M$, and $k\in S$. We call the buyers in the set  $B^{\mu}_{j}=\{i\in B_T : (i,j,k)\in\mu \textrm{ for some } k\in S_T\}$ and the sellers in the set $S^{\mu}_{j}=\{k\in S_T : (i,j,k)\in\mu \textrm{ for some } i\in B_T\}$ the \textit{partners} of middleman $j\in M_T$ in $T$-matching $\mu$. Let $M^{\mu}_{+}$ denote the set of those middlemen in $T$ who are involved in some trading triplet under $\mu$. Denote by $B^{\mu}_{0}$ ($S^{\mu}_{0}$) the set of those buyers (sellers) in coalition $T$, who are not partners of any middleman but are involved in some direct trade under $\mu$, as if they were partners of a fictitious middleman denoted by 0. Finally, denote the set of buyers, middlemen, and sellers in $T$ who are singletons in $\mu$ by $B^{\mu}_{s}$, $M^{\mu}_{s}$, and $S^{\mu}_{s}$ respectively. 

Obviously, $M^{\mu}_{s}$ together with the singletons $\{j\}$ $(j\in M^{\mu}_{+}=M_T\setminus M^{\mu}_{s})$ form a partition of $M_T$, $B^{\mu}_{s}$ together with the partner sets $B^{\mu}_{j}$ $(j\in M^{\mu}_{+}\cup \{0\})$ form a partition of $B_T$, and $S^{\mu}_{s}$ together with the partner sets $S^{\mu}_{j}$ $(j\in M^{\mu}_{+}\cup \{0\})$ form a partition of $S_T$. Moreover, the union of these three partitions form a partition of coalition $T$, called \textit{the $\mu$-induced partition} of $T$. Notice that $\mu$ induces a (complete) matching, denoted by $\mu^{(j)}$, between the partner sets $B^{\mu}_{j}$ and $S^{\mu}_{j}$ of each (real or fictitious) middleman $j\in M^{\mu}_{+}\cup \{0\}$. Consequently, $|B^{\mu}_{j}|=|S^{\mu}_{j}|$ for any $j\in M^{\mu}_{+}\cup \{0\}$. 

Note that a market $\gamma=(A,\hat{A})$ induces a TU game $v_\gamma$ where the worth of every coalition $T$ is  given by 
\begin{equation}\label{EqworthT}
    v_\gamma(T)=\max_{\mu\in \mathcal{A}(T)} \left[ \sum_{(i,k)\in \mu}a_{ik}+\sum_{(i,j,k)\in \mu} \hat{a}_{ijk}    
    = \sum_{j\in M^{\mu}_{+}\cup \{0\}} \sum_{(i,k)\in \mu^{(j)}} a^{(j)}_{ik} \right]
\end{equation}
Note from (\ref{EqworthT}) that all coalitions $T$ consisting of players of the same side (including singleton coalitions) are  worthless, that is, $v_\gamma(T)=0$.

A matching $\mu\in\mathcal{A}(T)$ will be called $T$-\textit{optimal} in the market $\gamma$ if $v_\gamma(T)=\sum\limits_{(i,k)\in \mu}a_{ik}+\sum\limits_{(i,j,k)\in \mu} \hat{a}_{ijk}$, that is, if $\mu$ solves the problem stated in (\ref{EqworthT}). Since $\mathcal{A}(T)$ is non-empty and finite, remark that there always exists (at least) one $T$-optimal matching in $\gamma$.
Given any $T\in 2^N$, we denote by $\mathcal{A}^*_\gamma(T)$  the set of $T$-optimal matchings in the market $\gamma$. We call \textit{optimal matching}   any $N$-optimal matching in $\gamma$. 

The following example illustrates the notions developed in this section.

\begin{example}
\label{Ex:Illustritaive}
Consider a market with middlemen $\gamma=(B,M,S,A,\hat{A})$ where
$B=\{b_{1}, b_{2}\}$, $M=\{m_{1}, m_{2}\}$, and $S=\{s_{1},s_{2}\}$ are the set of buyers, the set of middlemen, and the set of sellers, respectively. The total surplus of those basic coalitions formed by a pair of buyer and seller is given by the following two-dimensional matrix $A=(a_{ik})_{\substack{i\in B \\ k\in S}}$:
$$
A=\bordermatrix{~ & s_{1} & s_{2} \cr
                  b_{1} & 3 & 2 \cr
                  b_{2} & 1 & 5 \cr}
$$
and joint surplus generated by triplets formed by a buyer, a middleman, and a seller is given by the following three-dimensional matrix $\hat{A}=(\hat{a}_{ijk})_{\substack{i\in B\\ j\in M \\k\in S}}$:
$$\hat{A}=\begin{array}{cc}
 & \begin{array}{cc} s_{1} & s_{2}\end{array} \\
 \begin{array}{c} b_{1} \\ b_{2} \end{array} &
 \left(\begin{array}{cc} 4\; &\; 3 \\ 3\; &\; 5 \end{array}\right)\\
  & m_{1} \end{array}\qquad
 \begin{array}{cc}
 & \begin{array}{cc} s_{1} & s_{2} \end{array}\\
 \begin{array}{c} b_{1} \\ b_{2} \end{array} &
 \left(\begin{array}{cc} 6 \; &\; 2 \\ 2\; & \; 6 \end{array}\right)\\
  & m_{2} \end{array}.$$
Notice first that for a buyer-seller pair $(i,k)\in B\times S$, the total surplus of a trade with a middleman is at least as good as a direct trade. For instance, consider the buyer-seller pair $(b_{1},s_{2})\in B\times S$. They generate a total surplus of $2=a_{12}$ whereas they generate a strictly greater total surplus if the trade is mediated by middleman $m_{1}$: $\hat{a}_{112}=3>2=a_{12}$, and the same amount of total surplus is generated if middleman $m_{2}$ mediates the trade between them: $\hat{a}_{122}=2=a_{12}$.

Next, consider the set of all agents $N=B\cup M\cup S$ and two collections of basic coalitions $\mu=\{\{b_{1},m_{2},s_{1}\},\{b_{2},m_{2},s_{2}\},\{m_{1}\}\}$ and $\mu'=\{\{b_{1},m_{1},s_{1}\},\{b_{2},m_{2},s_{2}\}\}$. In the collection $\mu$, each buyer and seller belong to exactly one basic coalition whereas middleman $m_{2}$ appears in two distinctive basic coalitions and middleman $m_{1}$ appears as a singleton, and hence $\mu$ is a $N$-matching. 
Under $\mu$, $B^{\mu}_{m_2}=\{b_1,b_2\}$ are the buyer partners of $m_2$ and $S^{\mu}_{m_2}=\{s_1,s_2\}$ are the seller partners of $m_2$, whereas $m_1$ has no partners in $\mu$. Since all buyers and sellers are partners of middlemen, the induced partitions are $B^{\mu}_{s} \cup B^{\mu}_{0} \cup B^{\mu}_{m_1} \cup B^{\mu}_{m_2}=\emptyset \cup \emptyset \cup \emptyset \cup B_N =B_N=B$ for the buyers, $S^{\mu}_{s} \cup S^{\mu}_{0} \cup S^{\mu}_{m_1} \cup S^{\mu}_{m_2}=\emptyset \cup \emptyset \cup \emptyset \cup S_N=S_N=S$ for the sellers, and $M^{\mu}_{s} \cup M^{\mu}_{+}= \{m_1\} \cup \{m_2\}=M_N=M$ for the middlemen. 
In the collection $\mu'$ all agents, even the middlemen, belong to exactly one basic coalition, hence $\mu'$ is also a $N$-matching. It induces the partitions $B^{\mu}_{s} \cup B^{\mu'}_{0} \cup B^{\mu'}_{m_1} \cup B^{\mu'}_{m_2}=\emptyset \cup \emptyset \cup \{b_1\} \cup \{b_2\} =B_N=B$ of the buyers, $S^{\mu'}_{s} \cup S^{\mu'}_{0} \cup S^{\mu'}_{m_1} \cup S^{\mu'}_{m_2}=\emptyset \cup \emptyset \cup \{s_1\} \cup \{s_2\} =S_N=S$ of the sellers, and $M^{\mu}_{s} \cup M^{\mu}_{+}= \emptyset \cup \{m_1,m_2\}=M_N=M$ of the middlemen. 

Finally, let us consider the TU game $v_{\gamma}$ associated with the market $\gamma$. Consider the coalition $T=\{b_{1},m_{1},s_{1},s_{2}\}$. Then, the worth of the coalition $T$ is obtained by maximizing, over all possible $T$-matchings, the total value of basic coalitions in a matching. By ignoring the 0 value of the non-basic coalitions, $v_{\gamma}(T)= \max\{a_{11},a_{12},a_{111},a_{112}\}= \max\{3,2,4,3\}=4$. The optimal $T$-matching is $\mu^T=\left\{\{b_{1},m_{1},s_{1}\},\{s_{2}\}\right\}$. It induces the partition $B^{\mu^T}_{s} \cup B^{\mu^T}_{m_1} =\emptyset \cup \{b_1\} =B_T$ of the set of buyers and the partition $S^{\mu^T}_{s} \cup S^{\mu^T}_{m_1} =\{s_2\} \cup \{s_1\} =S_T$ of the set of sellers in coalition $T$. 

Now, consider again the grand coalition, $N$. The sum of the value of basic coalitions under the matching $\mu'$ is $a_{111}+a_{222}=4+6=10$, whereas under the matching $\mu$ it is equal to $12=6+6=a_{121}+a_{222}$. It is easily checked that the worth of $N$, $v_{\gamma}(N)$, is obtained under the matching $\mu$ which maximizes (\ref{EqworthT}), thus, the matching $\mu$ is an optimal matching. It induces matchings in the two-dimensional layer matrices: $\mu^{(0)}=\emptyset$ on the direct trade matrix $A^{(0)}=A$, $\mu^{(m_1)}=\emptyset$ on the layer matrix $A^{(m_1)}$ related to the unpartnered middleman $m_1\in M^{\mu}_{s}$, and $\mu^{(m_2)}=\left\{\{b_{1},s_{1}\},\{b_{2},s_{2}\}\right\}$ on the layer matrix $A^{(m_2)}$ related to the partnered middleman $m_2\in M^{\mu}_{+}$. The value of $\mu$ equals the sum of the values of these induced matchings in the corresponding layer matrices,  $12=0+0+(6+6)$. 
\end{example}

\medskip
In the next section we examine the core of the TU game $v_\gamma$ associated with the matching market with middlemen $\gamma$. We will show in particular that this game, called \textit{middlemen market game}, $v_\gamma$ is always totally balanced.

\section{The core of a market with middlemen}
\label{sec:core}
Given any market with middlemen $\gamma=(B,M,S, A, \hat{A})$, one can define the matrix $A^*=(a^*_{ik})_{\substack{i\in B \\ k\in S}}$ by 
\begin{equation}\label{EqAstar}
    a^*_{ik}=\max_{j\in M} \hat{a}_{ijk}, ~\forall (i,k)\in B\times S.
\end{equation} 
Note from (\ref{Assmatrices}) and (\ref{EqAstar}) that $a^*_{ik}$ gives the highest surplus possible in a trade involving buyer $i$ and seller $k$. Moreover, for all $(i,k)\in B\times S$,  we will use the notation  $m(i,k)$ to refer to the \textit{lowest-label middleman} $j\in M$ such that $a^*_{ik}=\hat{a}_{ijk}$, that is to say, $m(i,k)=\min \argmax\limits_{j\in M} \hat{a}_{ijk}=\min\{j\in M : \hat{a}_{ijk}=\max_{h\in M} \hat{a}_{ihk}\}$. Note that due to our simplifying assumption (\ref{Assmatrices}), $m(i,k)\in M$ is well defined for any buyer-seller pair $(i,k)\in B\times S$.

Thus, for any market $\gamma=(A, \hat{A})$, one can define the standard (two-sided) assignment market $\gamma^*=(B,S,A^*)$, where $A^*$ is given by (\ref{EqAstar}). Note that a matching $\nu$ in $\gamma^*$ is a partition of $B\cup S$ into singletons and mixed pairs $\{i,k\}$ such that $i\in B, k\in S$. We write $\nu(t)=t$ for all $t\in B\cup S$ such that  $\{t\}\in \nu$. In addition, we write $\nu(i)=k$ and $\nu(k)=i$ for all $(i,k)\in B\times S$ such that $\{i,k\}\in \nu$.
A matching $\nu$ in $\gamma^*$ is optimal if $\sum\limits_{\{i,k\}\in \nu}a^*_{ik}\geq \sum\limits_{\{i,k\}\in \nu'}a^*_{ik}$, for all matchings  $\nu'$ in  $\gamma^*$.

We connect the matchings, in particular the optimal matchings, of the respective markets $\gamma$ and $\gamma^*$. For expositional simplicity, we only consider the grand coalition, the concepts are analogously defined and the statements are straightforwardly derived for any subcoalition.

Let $\mu$ be a matching in a market with middlemen $\gamma=(B,M,S,A,\hat{A})$. As it induces partitions of the set of buyers $B$ and sellers $S$, and matchings $\mu^{(j)}$ ($j\in M\cup\{0\}$) between the partner sets for each middleman which are pairwise disjoint for different middlemen, the union $\bigcup_{j\in M\cup\{0\}} \mu^{(j)}$ augmented with the singletons in $B^{\mu}_{s}$ and $S^{\mu}_{s}$ defines a matching between $B$ and $S$. We denote it by $\mu^*$. The value of $\mu$ in the market $\gamma$ is clearly less than or equal to the value of $\mu^*$ in the two-sided market $\gamma^*$, that is, 
\begin{equation}\label{EqMuMuStar}
\mu_\gamma(B\cup M\cup S)=\sum\limits_{(i,k)\in \mu} a_{ik} + \sum\limits_{(i,j,k)\in \mu} \hat{a}_{ijk} \leq \sum\limits_{(i,k)\in \mu^*} a^*_{ik} = \mu^*_{\gamma^*}(B\cup S).
\end{equation}

Conversely, if $\sigma$ is a matching for the two-sided market $\gamma^*$, then $\sigma^{\triangle}=\{(i,m(i,k), k): (i,k)\in \sigma\} \cup 
\left\{ \{t\} \in \sigma \right\} \cup
\left\{ \{j\}:j\in M \mbox{ s.t. } j\neq m(i,k), \forall (i,k) \in \sigma\right \}$ is a matching for the market with middlemen $\gamma$.
The value of $\sigma^{\triangle}$ in the market $\gamma$ is clearly the same as the value of $\sigma$ in the two-sided market $\gamma^*$, that is 
\begin{equation}\label{EqSigmaStarSigma}
\sigma_{\gamma^*}(B\cup S) = \sum\limits_{(i,k)\in \sigma} a^*_{ik} = \sum\limits_{(i,m(i,k),k)\in \sigma^{\triangle}} \hat{a}_{im(i,k)k} = \sigma^{\triangle}_\gamma(B\cup M\cup S).
\end{equation}

Based on (\ref{EqMuMuStar}) and (\ref{EqSigmaStarSigma}), we derive the following relations between the optimal matchings and the optimum total surplusses in the two markets.
\begin{proposition}
\label{prop:opt_matchings}
Let $\gamma=(B,M,S,A,\hat{A})$ be a market with middlemen and $\gamma^*=(B,S,A^*)$ be the associated two-sided assignment market. Then
\\ (1) if $\sigma$ is an optimal matching for $\gamma^*$ then $\sigma^{\triangle}$ is an optimal matching for $\gamma$;
\\ (2) if $\mu$ is an optimal matching for $\gamma$ then $\mu^*$ is an optimal matching for $\gamma^*$. 
\\ Moreover, the optimum values of the two markets are the same.  
\end{proposition}

\begin{proof}
First, it follows from (\ref{EqMuMuStar}) that the optimum value of the three-sided market $\gamma$ is less than or equal to the optimum value of the associated two-sided market $\gamma^*$.

To see that the two market optimums coincide, let 
$\sigma$ be an optimal matching for the two-sided market $\gamma^*$. Then, by (\ref{EqSigmaStarSigma}), (\ref{EqMuMuStar}), and the optimality of $\sigma$, we get respectively, $\sigma_{\gamma^*}(B\cup S) = \sigma^{\triangle}_\gamma(B\cup M\cup S) \leq (\sigma^{\triangle})^*_{\gamma^*}(B\cup S) \leq \sigma_{\gamma^*}(B\cup S)$.
Thus, both inequalities must hold as equalities, implying that in the three-sided market $\gamma$, the matching $\sigma^{\triangle}$ attains the optimum value of the two-sided market $\gamma^*$ that, as observed above, is an upper bound for the optimum value of the three-sided market $\gamma$. Therefore, $\sigma^{\triangle}$ is an optimal matching for $\gamma$, proving claim (1) and the coincidence of the two market optimum values. 

To show claim (2), let $\mu$ be an optimal matching for $\gamma$. Then, by (\ref{EqMuMuStar}), (\ref{EqSigmaStarSigma}), and the optimality of $\mu$ in the three-sided market, $\mu_\gamma(B\cup M\cup S) \leq \mu^*_{\gamma^*}(B\cup S) = (\mu^*)^{\triangle}_{\gamma}(B\cup M\cup S) \leq \mu_\gamma(B\cup M\cup S)$.
Thus, both inequalities must hold as equalities, implying that in the two-sided market $\gamma^*$, the matching $\mu^*$ attains the optimum value of the three-sided market $\gamma$, that, as proved above, equals the optimum value of the two-sided market $\gamma^*$. Therefore, $\mu^*$ is an optimal matching for $\gamma^*$, proving claim (2).
\end{proof}

Proposition \ref{prop:opt_matchings} shows that one can always construct an optimal matching in the market with middleman $\gamma$ by first finding an optimal matching of the associated  two-sided market $\gamma^*$. Next, we reconsider Example \ref{Ex:Illustritaive} to illustrate Proposition \ref{prop:opt_matchings}.
\begin{example}[Example 1 Revisited]
Recall that, for the market $\gamma$, the total surplus of those basic coalitions formed by a pair of buyer and seller is given by the following two-dimensional matrix $A=(a_{ik})_{\substack{i\in B \\ k\in S}}$:
$$
A=\bordermatrix{~ & s_{1} & s_{2} \cr
                  b_{1} & 3 & 2 \cr
                  b_{2} & 1 & 5 \cr}
$$
and joint surplus generated by triplets formed by a buyer, a middleman, and a seller is given by the following three-dimensional matrix $\hat{A}=(\hat{a}_{ijk})_{\substack{i\in B\\ j\in M \\k\in S}}$:
$$\hat{A}=\begin{array}{cc}
 & \begin{array}{cc} s_{1} & s_{2}\end{array} \\
 \begin{array}{c} b_{1} \\ b_{2} \end{array} &
 \left(\begin{array}{cc} 4\; &\; 3 \\ 3\; &\; 5 \end{array}\right)\\
  & m_{1} \end{array}\qquad
 \begin{array}{cc}
 & \begin{array}{cc} s_{1} & s_{2} \end{array}\\
 \begin{array}{c} b_{1} \\ b_{2} \end{array} &
 \left(\begin{array}{cc} 6 \; &\; 2 \\ 2\; & \; 6 \end{array}\right)\\
  & m_{2} \end{array}.$$
First, we construct the associated two-sided market $(B,S,A^{*})$ where the set of buyers and the set of sellers are the same, and $A^{*}=(a_{ik}^{*})_{\substack{i\in B \\ k\in S}}$ is the valuation matrix defined by $a^{*}_{ik}=\max_{j\in M} \hat{a}_{ijk}, ~\forall (i,k)\in B\times S$. For instance, $a^{*}_{11}=\max \{\hat{a}_{111},\hat{a}_{121}\}=\max\{4,6\}=6=\hat{a}_{121}$. Then, $A^{*}=(a^{*}_{ik})_{\substack{i\in B \\ k\in S}}$ is
$$
A^{*}=\bordermatrix{~ & s_{1} & s_{2} \cr
                  b_{1} & 6 & 3 \cr
                  b_{2} & 3 & 6 \cr},
$$ 
where $a^{*}_{11}=\hat{a}_{121}$, $a^{*}_{12}=\hat{a}_{112}$, $a^{*}_{21}=\hat{a}_{211}$, and $a^{*}_{22}=\hat{a}_{222}$. 
Thus, $m(1,1)=m(2,2)=m_2$ and $m(1,2)=m(2,1)=m_1$. 

Notice that $\sigma=\left\{\{b_1,s_1\},\{b_2,s_2\}\right\}$ is the unique optimal matching in $\gamma^{*}$. 
Following Proposition \ref{prop:opt_matchings}, we construct the matching $\sigma^{\triangle}$ for the market $\gamma$: $\{b_{1},m_{2},s_{1}\}\in\sigma^{\triangle}$, $\{b_{2},m_{2},s_{2}\}\in\sigma^{\triangle}$,  
and $\{m_{1}\}\in\sigma^{\triangle}$ since there does not exists a pair $\{i,k\}\in\sigma$ such that $m_{1}=m(i,k)$. By Proposition \ref{prop:opt_matchings}, the matching $\sigma^{\triangle}=\left\{\{b_{1},m_{2},s_{1}\},\{b_{2},m_{2},s_{2}\},\{m_{1}\}\right\}$ thus obtained is optimal in $\gamma$ -- which was already known from our calculations in Example \ref{Ex:Illustritaive}.
Finally, the optimum values in the two markets are equal: $\sigma^{\triangle}_{\gamma}(B\cup M\cup S)=12=\sigma_{\gamma^*}(B\cup S)$.
\end{example}

\medskip
The following result proves that the TU game associated with a market with middlemen is always totally balanced.
\begin{theorem}\label{TheoremCoreNonempty}
Let $\gamma=(B,M,S,A,\hat{A})$ be a market with middlemen.
Then the associated middlemen matching market game $v_\gamma$ is totally balanced.
Moreover, 
\begin{equation}\label{core-0-facet}
\left\{ (x;y;z)\in Core(v_\gamma) : y=0 \right\} = \left\{ (x;0;z)\in \mathbb{R}^B\times\mathbb{R}^M\times\mathbb{R}^S : (x;z)\in Core(w_{\gamma^*}) \right\}  
\end{equation}
that is, the facet of $Core(v_\gamma)$ where all middlemen receive zero payoff is ``essentially the same'' as the core of the two-sided assignment game $w_{\gamma^*}$ induced by the two-sided assignment market $\gamma^*=(B,S,A^*)$.
\end{theorem}

\begin{proof}

First we show that the middlemen matching market game $v_\gamma$ is balanced for any matching market with middlemen $\gamma$, by showing the relation $\supseteq$ between the two payoff sets in (\ref{core-0-facet}) and observing that the set on the right is non-empty due to the balancedness of assignment games \citep{ss71}.

To this end, let $(x;z)\in Core(w_{\gamma^*})$ be arbitrary, but fixed. Then $\sum_{i\in B} x_i + \sum_{k\in S} z_k = w_{\gamma^*}(B\cup S) = v_{\gamma}(B\cup M\cup S)$, because by Proposition~\ref{prop:opt_matchings}, the optimum values in the two markets, hence, the grand coalition values in the two associated games are the same. Thus, the augmented payoff vector $(x;0;z)$ is efficient in $v_{\gamma}$. To see its coalitional rationality, let $T\subseteq B\cup M\cup S$ be arbitrary, but fixed. Let $\mu$ be an optimal $T$-matching in $\gamma$. By Proposition~\ref{prop:opt_matchings}, the value of the related two-sided matching $\mu^*$ between $B_T$ and $S_T$ in $\gamma^*$ is at least $v_{\gamma}(T)$. We get $v_{\gamma}(T) \leq \sum_{(i,k)\in \mu^*} a^*_{ik} \leq \sum_{(i,k)\in \mu^*} (x_i+z_k) \leq  \sum_{i\in B_T} x_i + \sum_{k\in S_T} z_k =(x;0;z)(T)$, where the last two inequalities come from the coalitional rationality of core payoff $(x;z)\in Core(w_{\gamma^*})$. Therefore, the augmented payoff vector $(x;0;z)$ is in the core of $v_{\gamma}$.   

To show the reverse inclusion $\subseteq$ in (\ref{core-0-facet}), take any payoff vector of the form $(x;0;z)$ from $Core(v_\gamma)$. As we proved above, such payoff vectors exist.  
By Proposition~\ref{prop:opt_matchings}, $(x;z)(B\cup S) = (x;0;z)(B\cup M\cup S) = v_{\gamma}(B\cup M\cup S) = w_{\gamma^*}(B\cup S)$, thus, the restricted payoff vector $(x;z)$ is efficient in $w_{\gamma^*}$. To see its coalitional rationality, let $R\subseteq B\cup S$ be arbitrary, but fixed. Let $\sigma$ be an optimal $R$-matching in $\gamma^*$. By Proposition~\ref{prop:opt_matchings}, the value of the related three-sided matching $\sigma^{\triangle}$ equals $w_{\gamma^*}(R)$. We get $w_{\gamma^*}(R) = \sum_{(i,m(i,k),k)\in\sigma^{\triangle}} a^{(m(i,k))}_{ik} \leq \sum_{(i,m(i,k),k)\in\sigma^{\triangle}} (x_i+0+z_k) \leq  \sum_{i\in B_R} x_i + \sum_{k\in S_R} z_k = (x;z)(R)$, where the last two inequalities come from the coalitional rationality of core payoff $(x;0;z)\in Core(v_{\gamma})$. Therefore, the restricted payoff vector $(x;z)$ is in the core of $w_{\gamma^*}$.

Finally, the total balancedness of a middlemen matching market game straightforwardly follows from the observation that the submarket obtained by restricting the surplus matrices to agents in a subcoalition induces precisely the subgame related to that subcoalition.  
\end{proof}
 
Any two-sided assignment market exhibits two distinguished core allocations, namely the buyer-optimal allocation and the  seller-optimal allocation. Under the  buyer-optimal (seller-optimal) allocation all buyers (sellers) simultaneously achieve their maximum core payoff and all sellers (buyers) simultaneously achieve their minimum core payoff. \cite{ss71} showed that, in the two-sided assignment market, the existence of these optimal allocations is a result of the ``lattice structure'' of the core. Moreover, \cite{d82} and \cite{l83} showed that the buyers (sellers) achieve their marginal contribution to the grand coalition at the buyer-optimal (seller-optimal) allocation. Nevertheless, this property does not extend to the arbitrary multi-sided markets. \cite{an19} study a special case of multi-sided markets where each of the $m$ sides has an optimal core allocation under which all agents of that side achieve their marginal contribution.

Our next result states that, in a matching market with middlemen, there exists an optimal core allocation for buyers (sellers): under their optimal allocation, all buyers (sellers) simultaneously achieve their marginal contribution. 

\begin{theorem}\label{TheoremBuyOptimal}
Let $\gamma=(A,\hat{A})$ be a  market with middlemen. Then the following statements hold.
\begin{enumerate}[(i)]
\item  There exists (a buyer-optimal core allocation) $x^B\in Core (v_\gamma)$ such that $x^B_i=mc_i(v_\gamma)$, for all $i\in B$. 
\item There exists (a seller-optimal core allocation) $x^S\in Core (v_\gamma)$ such that $x^S_k=mc_k(v_\gamma)$, for all $k\in S$. 
\end{enumerate}
\end{theorem}

\begin{proof}
It is shown in the proof of Theorem \ref{TheoremCoreNonempty} that, if an allocation $x^*\in \mathbb{R}_+^{B\cup S}$  is in the core of $v_{\gamma^*}$, then the augmented allocation $(x^*,0_{M})\in \mathbb{R}^N_+$ is in the core of the original market $v_{\gamma}$. We will use this fact twice in the proof of Theorem \ref{TheoremBuyOptimal}. 

\textbf{(i)} It is known from \cite{ss71} that there exists a buyer-optimal core allocation $y^B$  in the two-sided market $\gamma^*$, that is to say, $y^B_i=mc_i(v_{\gamma^*})$ for all $i\in B$. As noted above, we have $x^B =(y^B,0_{M})\in Core(v_{\gamma})$. 
 It thus remains to see that, for all $i\in B$,  $x^B_i= y^B_i=mc_i(v_{\gamma})$. 
 Indeed, combining (\ref{EqworthT}) and Proposition \ref{prop:opt_matchings}, note that $mc_i(v_{\gamma})= v_{\gamma}(N)-v_{\gamma}(N\setminus i)=v_{\gamma^*}(N)- v_{\gamma^*}(N\setminus i)$. Thus, we have $x^B_i= y^B_i=mc_i(v_{\gamma})=mc_i(v_{\gamma^*})$ for all $i\in B$. 
 
\textbf{(ii)} Letting $y^S$ be the seller-optimal core allocation in the two-sided market $\gamma^*$ and defining $x^S=(y^S,0_M)$, the same argument allows to write $x^S_k= y^S_k=mc_k(v_{\gamma})=mc_k(v_{\gamma^*})$ for all $k\in S$.
\end{proof}

Remark that Theorem \ref{TheoremBuyOptimal} claims the existence of a buyer-optimal (seller-optimal) allocation, but not that of a middleman optimal allocation. Indeed, it is not true in general that there exists an allocation where all middlemen achieve their highest payoff in the core.  The following example illustrates this point.

\begin{example}\label{ex:no_mc_middlemen}
Consider a market with middlemen $\gamma=(B,M,S,A,\hat{A})$ where $B=\{b_{1},b_{2}\}$, $M=\{1,2,3\}$, and $S=\{s_{1},s_{2}\}$ are the set of buyers, the set of middlemen, and the set of sellers, respectively.
Let the surplus of any direct buyer-seller trade be zero, that is, $A=(a_{ik}=0)_{i\in B, k\in S}$, and the surplus (layer) matrices of the mediated trades for the three middlemen be the following:
$$
A^{(1)}=\bordermatrix{~ & s_{1} & s_{2} \cr
                  b_{1} & 2 & 0 \cr
                  b_{2} & 0 & 0 \cr},
\quad
A^{(2)}=\bordermatrix{~ & s_{1} & s_{2} \cr
                  b_{1} & 0 & 0 \cr
                  b_{2} & 0 & 10 \cr},
\quad
A^{(3)}=\bordermatrix{~ & s_{1} & s_{2} \cr
                  b_{1} & 0 & 2 \cr
                  b_{2} & 4 & 0 \cr}.
$$
The associated two-sided assignment market $\gamma^{*}$ is obtained from the entrywise maximum surplusses:
$$
A^{*}=\bordermatrix{~ & s_{1} & s_{2} \cr
                  b_{1} & 2 & 2 \cr
                  b_{2} & 4 & 10 \cr}.
$$
The best mediators for the possible trading pairs are $m(b_{1},s_{1})=1$, $m(b_{2},s_{2})=2$, and $m(b_{1},s_{2})=m(b_{2},s_{1})=3$.

Notice that $\sigma=\{(b_{1},s_{1}),(b_{2},s_{2})\}$ is the (unique) optimal matching in the assignment market $\gamma^{*}$ with optimum value $12=2+10$.   
By Proposition~\ref{prop:opt_matchings}, the matching $\sigma^{\triangle}=\{(b_{1},1,s_{1}),(b_{2},2,s_{2}),\{3\}\}$ is an optimal matching in the market with middlemen $\gamma$ with the same optimum value $12=2+10$. 
Observe that although middleman 3 mediates no trade in the optimal matching, her role cannot be neglected in finding the core payoff allocations since $a^{(3)}_{b_{1}s_{2}}> \max\{a^{(1)}_{b_{1}s_{2}}, a^{(2)}_{b_{1}s_{2}}\}$ and $a^{(3)}_{b_{2}s_{1}}> \max\{a^{(1)}_{b_{2}s_{1}}, a^{(2)}_{b_{2}s_{1}}\}$. 

Let $v_{\gamma}$ be the TU game associated with the market $\gamma$, and $w_{\gamma^*}$ be the assignment game induced by the two-sided market $\gamma^*$.
As an illustration for (\ref{core-0-facet}) in Theorem~\ref{TheoremCoreNonempty}, it is easily checked that $Core(w_{\gamma^*})$ has the following four extreme points (on the left), and the augmented payoff vectors (on the right) are precisely those extreme points of $Core(v_{\gamma})$ which allocate zero to all three middlemen:
$$
\bordermatrix{
(w_{\gamma^*}) & x_{b_1} & x_{b_2} & z_{s_1} & z_{s_2} \cr
~ & 2 & 10 & 0 & 0 \cr
~ & 2 &  4 & 0 & 6 \cr
~ & 0 &  8 & 2 & 2 \cr
~ & 0 &  2 & 2 & 8 \cr}
\qquad
\bordermatrix{
(v_{\gamma}) & x_{b_1} & x_{b_2} & y_{1} & y_{2} & y_{3} & z_{s_1} & z_{s_2} \cr
~ & 2 & 10 & 0 & 0 & 0 & 0 & 0 \cr
~ & 2 &  4 & 0 & 0 & 0 & 0 & 6 \cr
~ & 0 &  8 & 0 & 0 & 0 & 2 & 2 \cr
~ & 0 &  2 & 0 & 0 & 0 & 2 & 8 \cr}.
$$
The first payoff vectors are the buyer-optimal core allocations, whereas the fourth ones are the seller-optimal core allocations in the respective markets. It is easily checked that, indeed, $mc_{b_1}=2$, $mc_{b_2}=10$, $mc_{s_1}=2$, $mc_{s_2}=8$ in both markets.
This illustrates Theorem~\ref{TheoremBuyOptimal}.
  
Also easily checked that $Core(v_{\gamma})$ also contains the following four payoff vectors:
$$
\bordermatrix{
(v_{\gamma}) & x_{b_1} & x_{b_2} & y_{1} & y_{2} & y_{3} & z_{s_1} & z_{s_2} \cr
~ & 0 & 8 & 2 & 0 & 0 & 0 & 2 \cr
~ & 0 & 4 & 2 & 0 & 0 & 0 & 6 \cr
~ & 2 & 4 & 0 & 6 & 0 & 0 & 0 \cr
~ & 0 & 2 & 0 & 6 & 0 & 2 & 2 \cr}.
$$
In the first two, middleman 1 gets her marginal value $mc_{1}=v_{\gamma}(N)-v_{\gamma}(N\setminus\{1\})=12-10=2$, while in the last two, middleman 2 gets her marginal value $mc_{2}=v_{\gamma}(N)-v_{\gamma}(N\setminus\{2\})=12-6=6$. 
The marginal value of middleman 3 is zero, because she is unmatched in the optimal matching, hence receives zero payoff in any core allocation. 

On the other hand, there exists no core allocation $(x;y;z)\in Core(v_{\gamma})$ such that all middlemen  achieve their marginal contribution, i.e. $y_{1}=mc_{1}=2$, $y_{2}=mc_{2}=6$, and $x_{3}=mc_{3}=0$.
Indeed, $\max\{y_{1}+y_{2} : (x;y;z)\in Core(v_{\gamma})\} \leq v_{\gamma}(N)-v_{\gamma}(N\setminus\{1,2\})=12-6=6$, since $v_{\gamma}(N\setminus\{1,2\})=a^{(3)}_{b_{1}s_{2}}+a^{(3)}_{b_{2}s_{1}}=2+4=6$. We remark that this upper bound is achieved, for example, at core allocation $(x;y;z)=(0,4;\, 2,4,0;\, 0,2)$. 
Hence, there is no core allocation at which all middlemen can achieve their maximum core payoffs (which are their marginal contributions) simultaneously.
\end{example}

\section{Core and competitive equilibria}
\label{sec:CE}

The aim of this section is to study the relationship between core and competitive equilibria in matching markets with middlemen. \cite{g60} defines competitive equilibrium prices and proves their existence for any assignment problem (see also \citealp{ss71}). \cite{t10} extends the coincidence between core and competitive equilibria for the classical three-sided assignment markets where buyers are forced to acquire exactly one item of each type. In a similar fashion, \cite{aetal16} generalizes the equivalence result for the generalized three-sided assignment markets where buyers can buy at most one good of each type. In both extensions, the existence of a competitive equilibrium is guaranteed whenever the core is non-empty.

Consider any market with middlemen where the set of buyers is $B=\{b_1 , \ldots, b_I\}$, the set of middlemen is $M=\{m_1, \ldots, m_J\}$, and the set of sellers is $S=\{s_1,\ldots,s_K\}$. Assume that buyers and sellers trade through the competitive market with the presence of middlemen and agents in the market are price-takers. Each buyer $i\in B$ demands at most one unit of the good, each seller $k\in S$ offers one unit for sale (recall that units owned by different sellers may be heterogeneous). Assume that buyer $i$ values the good of seller $k$ at $h_{ik}$, and the production cost of the good for seller $k$ is $c_{k}$. If buyer $i$ and seller $k$ trade directly, the transaction (search) cost $t_{ik}$, is incurred by buyer $i$. If buyer $i$ instead hires middleman $j$  and ends up purchasing the object owned by seller $k$, then the transaction cost $t_{ijk}$ is incurred by buyer $i$; and middleman $j\in M$ incurs the mediation cost $c_{j}^{ik}$. 
From of our assumption that the search cost of a direct trade is lower than that of a mediated trade, it comes that   $t_{ik}\geq t_{ijk}+c_{j}^{ik}$. 

Let $p_{k}$ be the price demanded by seller $k$ for her unit; and assume that  middleman $j\in M$ charges a fee $p_{j}^{ik}$ to  buyer $i$ when the latter uses $j$'s services to purchase the unit owned by seller $k$. Note that middlemen need not charge the same fee for each possible buyer-seller trade. That is, it may happen that $p_{j}^{ik}\neq p_{j}^{i'k'}$ when $j$ is mediating the respective pairs $(i,k)$ and $(i',k')$ (with the possibility of having either $i=i'$ or $k=k'$).

If the transaction between buyer $i$ and seller $k$ is realized through middleman $j$, then the utility of buyer $i$ is given by $h_{ik}-t_{ijk}-p^{ik}_{j}-p_{k}$, the benefit of seller $k$ is $p_{k}-c_{k}$, and the benefit of middleman $j$ is $p^{ik}_{j}-c^{ik}_{j}$. Thus, the total surplus is $h_{ik}-t_{ijk}-p^{ik}_{j}-p_{k}+p^{ik}_{j}+p_{k}-c^{ik}_{j}-c_{k}=h_{ik}-t_{ijk}-c_{j}^{ik}-c_{k}$. If $h_{ik}-t_{ijk}-c_{j}^{ik}-c_{k}<0$, no transaction will be realized since a transaction will go through only if it gives a non-negative utility to each of the three agents  $i$,  $j$ and  $k$.
Thus, for all $(i,j,k)\in B\times M\times S$, let $\hat{a}_{ijk}=\max\{0,h_{ik}-t_{ijk}-c^{ik}_{j}-c_{k}\}$ denote the surplus generated when a transaction is realized between buyer $i$ and seller $k$ through middleman $j$.
Similarly, when the transaction is realized directly between buyer $i$ and seller $k$, the utility of buyer $i$ is $h_{ik}-t_{ik}-p_{k}$, the benefit of seller $k$ is $p_{k}-c_{k}$, and hence the total surplus is $h_{ik}-t_{ik}-p_{k}+p_{k}-c_{k}=h_{ik}-t_{ik}-c_{k}=a_{ik}$. If $h_{ik}-t_{ik}-c_{k}<0$, no transaction will be realized between buyer $i$ and seller $k$. 
Thus, for all $(i,k)\in B\times S$, let $a_{ik}=\max\{0,h_{ik}-t_{ik}-c_{k}\}$ denote the surplus generated when a transaction is realized directly between buyer $i$ and seller $k$. 
Hence, this  detailed market situation can be summarized by a tuple simply giving the set of buyers, the set of middlemen, the set of sellers, and the two matrices with generic terms  $a_{ik}$ and $\hat{a}_{ijk}$ defined above. That is to say, the TU game $(N, v_{\gamma})$ associated with this market is defined precisely by the characteristic function $v_\gamma$ given  in (\ref{EqworthT}).

We want to show that each core allocation can be obtained as the result of trading at competitive prices. To do so, we need some definitions allowing to introduce the notion of competitive price vector.
A price vector $p\in \mathbb{R}^{B\times M\times S}_{+}\times \mathbb{R}^{S}_{+}$ contains the specific, possibly differentiated prices, of the mediation services for each buyer-middleman-seller configuration as well as the undifferentiated prices of the goods.

Given a matching market with middlemen $\gamma$, a \emph{feasible price vector} is $p\in \mathbb{R}^{B\times M\times S}_{+}\times \mathbb{R}^{S}_{+}$ such that $p^{ik}_{j}\ge c^{ik}_{j}$ for all $j\in M$ and $p_{k}\ge c_{k}$ for all $k\in S$. The set of basic coalitions that contain buyer $i\in B$ is $\mathcal{B}^{i}=\{E\in\mathcal{B}^{N}\mid i\in E\}$. Let $w^{i}(E)=h_{ik}-t_{ijk}$ be the valuation of buyer $i$ for $E=\{i,j,k\}$ and $w^{i}(E)=h_{ik}-t_{ik}$ be the valuation of buyer $i$ for $E=\{i,k\}$. 
Observe the relation 
\begin{equation}\label{valuation-basic}
v_\gamma(E)=\max \left\{0,\, w^i(E)-c(E\setminus \{i\} \right\}
\end{equation}
for any basic coalition $E\in\mathcal{B}^{i}$ containing buyer $i$. 

Next, for each feasible price vector $p\in \mathbb{R}^{B\times M\times S}_{+}\times\mathbb{R}^{S}_{+}$ we introduce the \textit{demand set} of each buyer $i\in B$.
\begin{definition}
Let $\gamma=(B,M,S,A,\hat{A})$ be a matching market with middlemen. The \emph{demand set} of buyer $i\in B$ at a feasible price vector $p\in\mathbb{R}^{B\times M\times S}_{+}\times\mathbb{R}^{S}_{+}$ is
\begin{center}
$D_{i}(p)=\left\{E\in\mathcal{B}^{i}\mid \, w^{i}(E)-p(E\setminus \{i\})\geq w^{i}(E')-p(E'\setminus \{i\}) \text{ for all } E'\in\mathcal{B}^{{i}} \right\}$.
\end{center}
\end{definition}

Note that $D_{i}(p)$ describes the set of basic coalitions containing buyer $i$ that maximize the net valuation of buyer $i$ at prices $p$. Notice also that the demand set of a buyer $i \in B$ is always non-empty since $i$ can always demand $E=\{i\}$ with a net profit of 0. 

Given a $N$-matching $\mu$, we say that a middleman $j\in M$ is \emph{unassigned} (by $\mu$) if $\mu(j)=j$ and we say that a seller $k\in S$ is \emph{unassigned} (by $\mu)$ if there is no $i\in B$ such that $k= \mu(i)$. Now, we can introduce the notion of \textit{competitive equilibrium} for our model. The literature has adopted the approach of \cite{rs90} for the definition of a competitive equilibrium in matching markets. We adapt this definition to our context with three-sided matching with buyers, middlemen, and sellers.

\begin{definition}
\label{def:CE}
Given a matching market with middlemen $\gamma=(B,M,S, A, \hat{A})$, a pair $(p,\mu)$ composed of a price vector $p$ and an $N$-matching $\mu$ forms a {\em competitive equilibrium} if
\begin{enumerate}[(i)]
\item $p$ is a feasible price vector, i.e., $p\in \mathbb{R}^{B\times M\times S}_{+}\times\mathbb{R}^{S}_{+}$ such that $p^{ik}_{j}\ge c^{ik}_{j}$ for all $j\in M$ and $p_{k}\ge c_{k}$ for all $k\in S$,
\item for each buyer $i\in B$ and basic coalition $E\in\mathcal{B}^{{i}}$, if $E\in\mu$ then $E\in D_{i}(p)$,
\item for each middleman $j\in M$, if $j$ is unassigned by $\mu$, then $p_{j}^{ik}=c_{j}^{ik}$ for all buyer-seller pairs $(i,k)\in B\times S$,
\item for each seller $k\in S$, if $k$ is unassigned by $\mu$, then $p_{k}=c_{k}$.
\end{enumerate}
\end{definition}

Observe that a competitive equilibrium consists of a set of prices and an $N$-matching where each buyer maximizes her utility under the assignment of $N$-matching and prices. Moreover, middlemen and sellers are competitive, in the sense that no middleman mediates a trade unless she can charge a fee (service price) at least equal to her cost and no seller agrees to sell her good without receiving at least her cost. If a pair $(p,\mu)$ is a competitive equilibrium, then we say that the price vector $p$ is a \textit{competitive equilibrium price vector} and the $N$-matching $\mu$ is a \emph{compatible matching}. The corresponding payoff vector for a given pair $(p,\mu)$ is called \textit{competitive equilibrium payoff vector}. This payoff vector
is $(x(p,\mu),y(p,\mu),z(p,\mu))\in\mathbb{R}^{B}\times\mathbb{R}^{M}\times\mathbb{R}^{S}$, defined by
\begin{eqnarray*}
x_{i}(p,\mu)&=&w^{i}(E^{\mu(i)})-p(E^{\mu(i)}\setminus \{i\})\mbox{ where } i\in E^{\mu(i)}\in\mu  \quad\mbox{ for all }i\in B,\\
y_{j}(p,\mu)&=&\sum\limits_{\{i,j,k\}\in\mathcal{\mu}} p_{j}^{ik} - \sum\limits_{\{i,j,k\}\in\mathcal{\mu}} c_{j}^{ik} = p_{j}(\mu) - c_{j}(\mu) \quad \mbox{for all } j\in M,\\
z_{k}(p,\mu)&=&p_{k}-c_{k} \quad \mbox{ for all  } k\in S.
\end{eqnarray*}
Notice the dependence of the aggregated service prices (fees) $p_{j}(\mu)$ and the aggregated service costs $c_{j}(\mu)$ on the matching $\mu$.

We denote the set of competitive equilibrium payoff vectors of market $\gamma$ by $\mathcal{CE}(\gamma)$. We now study the relationship between the core of $\gamma=(B,M,S,A,\hat{A})$ and the set of competitive equilibrium payoff vectors. First, we show that an $N$-matching $\mu$ is an optimal matching whenever it constitutes a competitive equilibrium with a feasible price vector $p$.

\begin{lemma}
\label{lem:opt_match}
Given a matching market with middlemen $\gamma=(B,M,S,A,\hat{A})$, if $(p,\mu)$ is a competitive equilibrium, then $\mu$ is an optimal matching. 
\end{lemma}

\begin{proof}
Consider a competitive equilibrium $(p,\mu)$ and another $N$-matching $\mu'\in\mathcal{M}(B,M,S)$. 
For buyer $i\in B$, let $E^{\mu(i)}\in\mathcal{B}^{i}$ be the (unique) basic coalition assigned to $i$ under the matching $\mu$, that is, $i\in E^{\mu(i)}\in\mu$, and $E^{\mu'(i)}\in\mathcal{B}^{i}$ be the (unique) basic coalition assigned to $i$ under the matching $\mu'$, that is, $i\in E^{\mu'(i)}\in \mu'$. 
We can assume, without loss of generality, that $\mu'$ is such that for any $i\in B$, if $E^{\mu'(i)}$ is not a singleton then $w^{i}(E^{\mu'(i)})-c(E^{\mu'(i)}\setminus\{i\})\geq 0$, for otherwise we could replace $E^{\mu'(i)}$ with the singleton coalitions of its members and get a (finer) $N$-matching $\mu''$ with the same total value for $N$.  
Then,
\begin{align*}
&\sum\limits_{E\in\mu}v_{\gamma}(E) \overset{\text{\tiny{(1)}}}{\geq} \sum\limits_{i\in B} \left(w^{i}(E^{\mu(i)})-c(E^{\mu(i)}\setminus\{i\}) \right) \\
&\overset{\text{\tiny{(2)}}}{\geq} \sum\limits_{i\in B} \left( w^{i}(E^{\mu'(i)})-c(E^{\mu(i)}\setminus\{i\})-p(E^{\mu'(i)}\setminus\{i\})+p(E^{\mu(i)}\setminus\{i\}) \right)\\
 &\overset{\text{\tiny{(3)}}}{=} \sum\limits_{i\in B} \left(w^{i}(E^{\mu'(i)})-c(E^{\mu(i)}\setminus\{i\})\right) - p\left(\bigcup_{i\in B}E^{\mu'(i)}\setminus B\right)+p\left(\bigcup_{b_i\in B}E^{\mu(i)}\setminus B\right)\\
&\overset{\text{\tiny{(4)}}}{=} \sum\limits_{i\in B}w^{i}(E^{\mu'(i)})-c\left(\bigcup_{i\in B}E^{\mu(i)}\setminus B\right) - p\left(\left(\bigcup_{i\in B}E^{\mu'(i)}\setminus\bigcup_{i\in B}E^{\mu(i)}\right)\setminus B\right)\\
&\qquad +p\left(\left(\bigcup_{i\in B}E^{\mu(i)}\setminus\bigcup_{i\in B}E^{\mu'(i)}\right)\setminus B\right)\\
&\overset{\text{\tiny{(5)}}}{=} \sum\limits_{i\in B}w^{i}(E^{\mu'(i)})-c\left(\bigcup_{i\in B}E^{\mu(i)}\setminus B\right) - c\left(\left(\bigcup_{i\in B}E^{\mu'(i)}\setminus\bigcup_{i\in B}E^{\mu(i)}\right)\setminus B\right)\\
&\qquad +p\left(\left(\bigcup_{i\in B}E^{\mu(i)}\setminus\bigcup_{i\in B}E^{\mu'(i)}\right)\setminus B\right) \\
&\overset{\text{\tiny{(6)}}}{=} \sum\limits_{i\in B}w^{i}(E^{\mu'(i)})-c\left(\bigcup_{i\in B}E^{\mu'(i)}\setminus B\right) - c\left(\left(\bigcup_{i\in B}E^{\mu(i)}\setminus\bigcup_{i\in B}E^{\mu'(i)}\right)\setminus B\right)\\
&\qquad +p\left(\left(\bigcup_{i\in B}E^{\mu(i)}\setminus\bigcup_{i\in B}E^{\mu'(i)}\right)\setminus B\right)\\
&\overset{\text{\tiny{(7)}}}{\geq}  \sum\limits_{i\in B} \left(w^{i}(E^{\mu'(i)})-c(E^{\mu'(i)}\setminus\{i\})\right) \overset{\text{\tiny{(8)}}}{=} \sum\limits_{E\in\mu'}v_{\gamma}(E),
\end{align*}
where inequality $\overset{\text{\tiny{(1)}}}{\geq}$ follows from the relation $\displaystyle v_\gamma(E)=\max \left\{0,\, w^i(E)-c(E\setminus \{i\} \right\}$ for any basic coalition $E\in\mathcal{B}^{i}$, and inequality $\overset{\text{\tiny{(2)}}}{\geq}$ follows from the definition of the demand set and the fact that $(p,\mu)$ is a competitive equilibrium: $w^{i}(E^{\mu(i)})\geq w^{i}(E^{\mu'(i)})-p(E^{\mu'(i)}\setminus\{i\})+p(E^{\mu(i)}\setminus\{i\})$. 
Equality $\overset{\text{\tiny{(4)}}}{=}$ is the result of canceling out the common service prices, while 
equality $\overset{\text{\tiny{(5)}}}{=}$ follows from the fact that for all $j \in \left(\bigcup_{i\in B}E^{\mu'(i)}\setminus \bigcup_{i\in B}E^{\mu(i)}\right)\cap M$, $p^{ik}_{j}=c_{j}^{ik}$ and for all $k\in  \left(\bigcup_{i\in B}E^{\mu'(i)}\setminus \bigcup_{i\in B}E^{\mu(i)}\right)\cap S$, $p_{k}=c_{k}$.
Equality $\overset{\text{\tiny{(6)}}}{=}$ shows the rearrangement of costs incurred in the union of the two matchings, and inequality $\overset{\text{\tiny{(7)}}}{\geq}$ follows from the feasibility of the price vector $p$.
Finally, equality $\overset{\text{\tiny{(8)}}}{=}$ comes from relation (\ref{valuation-basic}) under our assumption on $\mu'$.
\end{proof}

Now, we can provide the main result of this section. We establish the equivalence between the core and the set of competitive equilibrium payoff vectors.
 
\begin{theorem}\label{thm:Core_euqivalence_CE}
Given a matching market with middlemen $\gamma=(B,M,S,A,\hat{A})$, the core of the market, $Core(\gamma)$, coincides with the set of competitive equilibrium payoff vectors, $\mathcal{CE}(\gamma)$.
\end{theorem}

\begin{proof}
First, we show that if $(p,\mu)$ is a competitive equilibrium, then its corresponding competitive equilibrium payoff vector $X=(x(p,\mu), y(p,\mu), z(p,\mu))\in\mathcal{CE}(\gamma)$ is a core element. Recall that $x_{i}(p,\mu)=w^{i}(E^{\mu(i)})-p(E^{\mu(i)}\setminus\{i\})$ for all buyers $i\in B$ where $i\in E^{\mu(i)}\in\mu$, $y_{j}(p,\mu)=\sum\limits_{\{i,j,k\}\in\mathcal{\mu}} p_{j}^{ik} - \sum\limits_{\{i,j,k\}\in\mathcal{\mu}} c_{j}^{ik} = p_{j}(\mu) - c_{j}(\mu)$ for all middlemen $j\in M$, and $z_{k}(p,\mu)=p_{k}-c_{k}$ for all sellers $k\in S$. Let us check that for all basic coalitions $E\in\mathcal{B}$ it holds $X(E)\geq v_{\gamma}(E)$. Notice that if $E$ does not contain any buyer $i\in B$, then $v_{\gamma}(E)=0$ and hence the core inequality trivially holds. Otherwise, take $E\in\mathcal{B}$ such that $i\in E$ for some $i\in B$. Again, if $v_{\gamma}(E)=0$, the core inequality trivially holds. Thus, assume $v_{\gamma}(E)>0$. Then,
\begin{align*}
X(E)&=w^{i}(E^{\mu(i)})-p(E^{\mu(i)}\setminus \{i\})+p(E\setminus \{i\})-c(E\setminus \{i\}) \\
&\geq w^{i}(E)-p(E\setminus \{i\})+p(E\setminus \{i\})-c(E\setminus \{i\}) \\
&=w^{i}(E)-c(E\setminus \{i\})=v_{\gamma}(E),
\end{align*}
where the inequality follows from the fact that $(p,\mu)$ is a competitive equilibrium, and the last equality comes from relation (\ref{valuation-basic}) under our assumption on the value of $E$. It remains to check that $X$ is efficient. Since at the matching $\mu$ each buyer $i\in B$ and each seller $k\in S$ appears in at most one buyer-seller pair or one buyer-middleman-seller triplet and each middleman can serve arbitrary number of buyer-seller pairs, we get
\begin{align*}
X(N)&=\sum\limits_{i\in B} \left[w^{i}(E^{\mu(i)})-p(E^{\mu(i)}\setminus\{i\})\right]+p(M\cup S)-c(M\cup S) \\
&=\sum\limits_{i\in B}\left[ w^{i}(E^{\mu(i)})-p(E^{\mu(i)}\setminus \{i\})+p(E^{\mu(i)}\setminus \{i\})-c(E^{\mu(i)}\setminus \{i\}) \right]\\
& \quad +\sum\limits_{j\notin\bigcup\limits_{i\in B}E^{\mu(i)}} \left(p_{j}(\mu)-c_{j}(\mu)\right)+\sum\limits_{k\notin\bigcup\limits_{i\in B}E^{\mu(i)}} (p_{k}-c_{k}) \\
&=\sum\limits_{i\in B}\left[w^{i}(E^{\mu(i)})-c(E^{\mu(i)}\setminus\{i\})\right] \\
&=\sum\limits_{i\in B}v_{\gamma}(E^{\mu(i)})=\sum\limits_{E\in\mu}v_{\gamma}(E),
\end{align*}
where the third equality holds since $p_{j}^{ik}=c_{j}^{ik}$ for unassigned middlemen $j\in M$ and $p_{k}=c_{k}$ for unassigned seller $s_{k}\in S$.
The fourth equality holds because of the optimality of $\mu$ by Lemma~\ref{lem:opt_match} and the observation that, as in any optimal matching, for any $i\in B$, we must have $w^{i}(E^{\mu(i)})-c(E^{\mu(i)}\setminus\{i\}) \geq 0$.  

We have shown that if $(p,\mu)$ is a competitive equilibrium, then its competitive equilibrium payoff vector $X\in \mathcal{CE}(\gamma)$ is a core allocation. Next, we show that the reverse implication holds. That is, if $X\in \mathbb{R}^{B}\times\mathbb{R}^{M}\times\mathbb{R}^{S}$ is a core allocation, then it is the payoff vector related to some competitive equilibrium $(p,\mu)$, where $\mu$ is any optimal matching and $p$ is a competitive equilibrium price vector. 

Let us define $p_{j}^{ik}=X_{j}+c_{j}^{ik}$ for all basic triplet $\{i,j,k\}\in\mathcal{B}$. Given any optimal matching $\mu$, take the aggregate service prices $p_{j}(\mu)=\sum\limits_{\{i,j,k\}\in\mu} p_{j}^{ik}$ for all middleman $j\in M$. Define $p_{k}=X_{k}+c_{k}$ for all sellers $k\in S$. Notice first that, since $X\in Core(\gamma)$, if seller $k$ is unassigned by the matching $\mu$, $p_{k}=X_{k}+c_{k}=c_{k}$ and if a middleman does not mediate a trade between a buyer-seller pair $(i,k)$ under the matching $\mu$, $p_{j}^{ik}=X_{j}+c_{j}^{ik}=c_{j}^{ik}$. Moreover, $X(E^{\mu(i)})=v_{\gamma}(E^{\mu(i)})$ for all $i\in B$ and $X(E')\ge v_{\gamma}(E')$ for all $E'\in\mathcal{B}^{i}$. 
Notice that by the optimality of $\mu$, $v_{\gamma}(E^{\mu(i)})=w^{i}(E^{\mu(i)})-c(E^{\mu(i)}\setminus\{i\})\geq 0$ for all $i\in B$.
Then, for all $i\in B$ and $E'\in \mathcal{B}^{i}$,

\begin{align*}
w^{i}(E^{\mu(i)})-p(E^{\mu(i)}\setminus\{i\})&=v_{\gamma}(E^{\mu(i)})+c(E^{\mu(i)}\setminus\{i\})-p(E^{\mu(i)})\setminus\{i\})\\
&=X(E^{\mu(i)})+c(E^{\mu(i)}\setminus \{i\})-p(E^{\mu(i)}\setminus\{i\})\\
&=X_{i}\\
&\geq v_{\gamma}(E')-X(E'\setminus\{i\})\\
&=v_{\gamma}(E')-\big[ p(E'\setminus\{i\})-c(E'\setminus \{i\})\big]\\
&\geq w^{i}(E')-p(E'\setminus\{i\})
\end{align*}
where the first inequality follows from the fact that $X\in Core(\gamma)$ and the second inequality comes from relation (\ref{valuation-basic}). This shows that $E^{\mu(i)}\in D_{i}(p)$ which concludes the proof. 
\end{proof}

We have shown that the core and the set of competitive equilibrium payoff vectors coincide under the assumption that middlemen need not charge the same price for two different buyer-seller trade. Next example shows that if we consider the case where middlemen charge a fixed price for each buyer-seller trade that they mediate, then a core allocation need not to be supported by competitive prices. 

\begin{example}
\label{ex:fixedprice_core_not_ce}
Consider a market with middlemen $\gamma = (B, M, S, A, \hat{A})$ where $B = \{b_1, b_2\}$, $M = \{m_1, m_2\}$, and $S = \{s_1, s_2\}$ are the set of buyers, the set of middlemen, and the set of sellers, respectively. The total surplus of those basic coalitions formed by a pair of buyer and seller is given by the following two-dimensional matrix $A=(a_{ik})_{\substack{i\in B \\ k\in S}}$:
$$
A=\bordermatrix{~ & s_{1} & s_{2} \cr
                  b_{1} & 3 & 2 \cr
                  b_{2} & 1 & 5 \cr}
$$
and joint surplus generated by triplets formed by a buyer, a middleman, and a seller is given by the following three-dimensional matrix $\hat{A}=(\hat{a}_{ijk})_{\substack{i\in B\\ j\in M \\k\in S}}$:
$$\hat{A}=\begin{array}{cc}
 & \begin{array}{cc} s_{1} & s_{2}\end{array} \\
 \begin{array}{c} b_{1} \\ b_{2} \end{array} &
 \left(\begin{array}{cc} 4\; &\; 3 \\ 3\; &\; 5 \end{array}\right)\\
  & m_{1} \end{array}\qquad
 \begin{array}{cc}
 & \begin{array}{cc} s_{1} & s_{2} \end{array}\\
 \begin{array}{c} b_{1} \\ b_{2} \end{array} &
 \left(\begin{array}{cc} 6 \; &\; 2 \\ 2\; & \; 6 \end{array}\right)\\
  & m_{2} \end{array}.$$
Notice first that there is a unique optimal matching, $\mu=\{(b_{1},m_{2},s_{1}),(b_{2},m_{2},s_{2})\}$. \footnote{Without loss of generality, let us denote buyers, middlemen, and sellers by their indices in an order of buyer-middleman-seller. For instance, $\{1,2,1\}$ for the coalition $\{b_{1},m_{2},s_{1}\}$.} If the service price of each possible trade was fixed for any middlemen, then under the optimal matching $\mu$, middleman $m_{2}$ would charge the same price for both trades she mediates, i.e. $p_{2}^{11}=p_{2}^{22}$. Suppose now that costs for sellers and middlemen assigned under the matching $\mu$ are equal to zero, $c_{1}=c_{2}=c_{2}^{11}=c_{2}^{22}=0$. Then, following ($\ref{valuation-basic}$), net valuations of buyer $b_{1}$ and buyer $b_{2}$ are $w^{1}(\{1,2,1\})=v(\{1,2,1\})=6$ and $w^{2}(\{2,2,2\})=v(\{2,2,2\})=6$.

Now, take the core allocation $X=(3,5;0,3;1,0)$. Since all costs are equal to zero, $z_{1}(p,\mu)=p_{1}-c_{1}=1$ implies that $p_{1}=1$ for seller $s_{1}$, $z_{2}(p,\mu)=0$ implies that $p_{2}=0$ for seller $s_{2}$. For middleman $m_{1}$, $y_{1}(p,\mu)=0$ since she is unassigned under the optimal matching $\mu$ whereas $y_{2}(p,\mu)=p^{11}_{2}-c_{2}^{11}+p^{22}_{2}-c_{2}^{22}=1.5-0+1.5-0=3$. Together with $p^{11}_{2}=1.5$, $p_{1}=1$ imply that, under the matching $\mu$, $w^{1}(\{1,2,1\})-p^{11}_{2}-p_{1}=6-1.5-1=3.5\neq 3=x_{1}(p,\mu)$ and $p^{22}_{2}=1.5$ together with $p_{2}=0$ imply that $w^{2}(\{2,2,2\})-p^{22}_{2}-p_{2}=6-1.5-0=4.5\neq 5=x_{2}(p,\mu)$. Hence, the core allocation $X$ is not supported by the competitive equilibrium $(p,\mu)$ when middlemen charge a fixed service price for each trade they mediate. 

Note that the core allocation $X=(3,5;0,3;1,0)$ can be supported by the competitive prices when middleman $m_{2}$ have different service prices for each trade she mediates under the matching $\mu$. Take $p_{2}^{11}=2$ and $p_{2}^{22}=1$, $p_{1}=1$, and $p_{2}=0$. Suppose again that costs for sellers and middlemen assigned under the matching $\mu$ are equal to zero, $c_{1}=c_{2}=c_{2}^{11}=c_{2}^{22}=0$. Then, following ($\ref{valuation-basic}$), $w^{1}(\{1,2,1\})=v(\{1,2,1\})=6$ and $w^{2}(\{2,2,2\})=v(\{2,2,2\})=6$. One can easily see that, $x_{1}(p,\mu)=w^{1}(\{1,2,1\})-p^{11}_{2}-p_{1}=6-2-1=3$, $x_{2}(p,\mu)=w^{2}(\{2,2,2\})-p^{22}_{2}-p_{2}=6-1-0=5$, $y_{1}(p,\mu)=0$ since she is unassigned under the matching $\mu$, $y_{2}(p,\mu)=p_{2}^{11}-c_{2}^{11}+p_{2}^{22}-c_{2}^{22}=2-0+1-0=3$, $z_{1}(p,\mu)=p_{1}-c_{1}=1-0=1$, $z_{2}(p,\mu)=p_{2}-c_{2}=0$.
\end{example}

\section{Concluding remarks}
\label{sec:conc}

We have considered a class of multi-sided matching markets where a trade between buyer-seller pairs can be realized with or without middlemen. We allow a middleman to serve the entire market by mediating as many trades as the size of the short side of the market while buyer-seller pairs can also trade directly. We have associated a classical two-sided assignment market with a matching market with middlemen by taking for each buyer-seller pair the maximum surplus that this pair can achieve with the costless help of middlemen. We have shown that the non-empty core of this associated two-sided assignment market can be embedded in the core of the matching market with middlemen by allocating zero payoff to all middlemen.

For these markets we have introduced an associated TU game, thereby we have extended the classical (two-sided) assignment markets of \cite{ss71} to a special multi-sided case. We have shown that every matching market with middlemen has a non-empty core. In addition, we have proved that there exists a buyer-optimal and a seller-optimal core allocation for every matching market with middlemen. Unlike in other extensions previously studied, it is shown that all buyers (sellers) achieve their marginal contribution simultaneously at the buyer-optimal (seller-optimal) core allocation. In addition, we have provided an example to show that it is not the case for middlemen: in general there does not exist an allocation that every middleman weakly prefers to any other allocation in the core. Finally, we have studied the relationship between the core and the set of competitive equilibria. We have established the coincidence between the core and the set of competitive equilibrium payoff vectors.

A possible direction for further research is to study the relationship between the core and another set-wise solution concept, the bargaining set. \cite{so99} proved the equivalence between the classical bargaining set of \cite{dm67}, a set-wise solution concept based on bargaining possibilities of players, for two-sided assignment games (see also \citealp{s08} for related results on other partitioning games). \cite{b21} generalized this result among others to a larger class known as (quasi)-hyperadditive games. For multi-sided matching markets, the coincidence result between the classical bargaining set and the core is exhibited only to the class of supplier-firm-buyer games (\citealp{as18}). Nevertheless, the methods used in the aforementioned papers do not seem to carry over to our model and we leave exploring the relationship between the bargaining set and the core for future research.

\bibliographystyle{apalike}

\bibliography{ABS}
\end{document}